\documentclass[11 pt]{article}
\usepackage{amssymb,amsthm,amsmath,color,graphics}
\theoremstyle{plain}
\newtheorem{theorem}{Theorem}
\newtheorem{lemma}[theorem]{Lemma}

\newtheorem{prop}[theorem]{Proposition}

\newtheorem{remark}[theorem]{Remark}

\theoremstyle{definition}
\newtheorem{definition}{Definition}

\newcommand{\myket}[1]{|#1\rangle}
\newcommand{\mybra}[1]{\langle #1 |}

\newcommand{\myketbra}[2]{|#1\rangle\langle #2 |}

\newcommand{\ket}[1]{\ensuremath{\left|#1\right>}}
\newcommand{\bra}[1]{\ensuremath{\left<#1\right|}}

\newcommand{\ketbra}[2]{| #1 \rangle\langle #2 |}

\newcommand{\be}{\begin{equation}}
\newcommand{\ee}{\end{equation}}
\newcommand{\bc}{\begin{center}}
\newcommand{\ec}{\end{center}}
\newcommand{\bea}{\begin{eqnarray}}
\newcommand{\eea}{\end{eqnarray}}

\newcommand{\tr}{\text{Tr}}

\renewcommand{\sp}{\text{sp}}

\newcommand{\W}{\mathcal{W}}

\def\opone{\leavevmode\hbox{\small1\kern-3.8pt\normalsize1}}

\newcommand{\sys}[1]{{#1}}

\newcommand{\setft}[1]{\mathrm{#1}}
\newcommand{\lin}[1]{\setft{L}\left(#1\right)}

\newcommand{\pos}[1]{\setft{Pos}\left(#1\right)}

\def\I{\mathbb{I}}
\def\({\left(}
\def\){\right)}
\def\X{\mathcal{X}}
\def\Y{\mathcal{Y}}
\def\Z{\mathcal{Z}}
\def\W{\mathcal{W}}

\begin{document}

\bc {\LARGE \bf Unconditionally-secure and reusable public-key
authentication
}

%
%
%
%

\medskip

\renewcommand{\thefootnote}{\fnsymbol{footnote}}
{\footnotesize L. M. Ioannou$^{1,2}$\footnote{lmi@iqc.ca} and M.
Mosca$^{1,2,3}$\footnote{mmosca@iqc.ca} \\ {\it
$^1$Institute for Quantum Computing, University of Waterloo,\\
200 University Avenue, Waterloo, Ontario, N2L 3G1, Canada \\
$^2$Department of Combinatorics and Optimization, University of Waterloo,\\
200 University Avenue, Waterloo, Ontario, N2L 3G1, Canada\\
$^3$Perimeter Institute for Theoretical Physics\\31 Caroline Street
North, Waterloo, Ontario, N2L 2Y5, Canada} }

\ec

\begin{quote}{\small
We present a quantum-public-key identification protocol and show
that it is secure against a computationally-unbounded adversary.
This demonstrates for the first time that unconditionally-secure and
reusable public-key authentication is possible in principle with
(pure-state) public keys.}
\end{quote}


\section{Introduction}

Public-key cryptography has proved to be an indispensable tool in
the modern information security infrastructure.  Most notably,
digital signature schemes form the backbone of Internet commerce,
allowing trust to be propagated across the network in an efficient
fashion.  In turn, public-key encryption allows the private
communication of messages (or, more usually, the establishment of
symmetric secret keys) among users who are authenticated via digital
signatures.  The security of these classical public-key
cryptosystems relies on assumptions on the difficulty of certain
mathematical problems \cite{MvOV96}.  Gottesman and Chuang
\cite{qphGC01} initiated the study of quantum-public-key
cryptography, where the public keys are quantum systems, with the
goal of obtaining the functionality and efficiency of public-key
cryptosystems but with information-theoretic security.  They
presented a secure one-time digital signature scheme for signing
classical messages, based on Lamport's classical scheme
\cite{Lam79}.

In a public-key framework, Alice chooses a random private key,
creates copies of the corresponding public key via some
publicly-known algorithm, and distributes the copies in an
authenticated fashion to all potential ``Bobs''.  In principle, this
asymmetric setup allows, e.g., any Bob to send encrypted messages to
Alice or to verify any signature for a message that Alice digitally
signed.  By eliminating the need for each Alice-Bob pair to
establish a secret key (in large networks where there may be many
``Alices'' and ``Bobs''), the framework vastly simplifies key
distribution, which is often the most costly part of any
cryptosystem, compared to a framework that uses only symmetric keys.

Some remarks about the quantum-public-key framework are in order.
First, we address the issue of \emph{purity} of the quantum public
keys. In principle, the quantum public key can be either in a pure
or mixed state from Alice's point of view (a mixed state is a fixed
probabilistic distribution of pure states). Gottesman and Chuang
\cite{qphGC01} assumed pure-state public keys. For digital signature
schemes, this purity is crucial; for, otherwise, Alice could cheat
by sending different public keys to different ``Bobs''. Purity
prevents Alice's cheating in this case because different ``Bobs''
can compare their copies of the public key via a ``distributed
\textsc{swap}-test'' \cite{qphGC01} to check they are the same (with
high probability), much like can be done in the case of classical
public keys.  But any scheme can benefit from an equality test,
since an adversary who tries to substitute bad keys for legitimate
ones could thus be caught.  There is no known equality test
guaranteed to recognize when two mixed states are equal.  Thus,
having mixed-state public keys seems to be at odds with what it
means to be ``public'', i.e., publicly verifiable.\footnote{Other
authors have defined the framework to include mixed public keys, and
Ref. \cite{KKNY05} proposes an encryption scheme with mixed public
keys that is reusable and unconditionally secure
\cite{HKK08}.}~~Even though the scheme we present in this paper does
not make explicit use of the ``distributed \textsc{swap}-test''
(because we assume the public keys have been securely distributed),
it can do so in principle.  We view this as analogous to how modern
public-key protocols do not specify use of an equality test among
unsure ``Bobs'', but how such a test is supported by the framework
to help thwart attempts to distribute fake keys.

Second, we address the issue of \emph{usability} of
quantum-public-key systems. The states of two quantum public keys
corresponding to two different private keys always have overlap less
than $(1-\delta)$, for some positive and publicly known $\delta$.
Thus, a striking aspect of the quantum-public-key framework is that
the number of copies of the public key in circulation must be
limited (if we want information-theoretic security). If this were
not the case, then an adversary could collect an arbitrarily large
number of copies, measure them all, and determine the private key.
 By adjusting protocol parameters, this limit on the number of copies
of the quantum public key can be increased in order to accommodate
more users (or uses; see next paragraph for a discussion on
``reusability''). Thus, in practice, there is no restriction on the
usability of a quantum-public-key system as long as an accurate
estimate can be made of the maximum number of users/uses.

Presumably, adjusting the protocol parameters (as discussed above)
in order to increase the maximum number of copies of the quantum
public key in circulation would result in a less efficient protocol
instance, and this is one kind of tradeoff between efficiency and
usability in the quantum-public-key setting. Another kind concerns
\emph{reusability}. The abovementioned digital signature scheme is
``one-time'' because only one message may be signed under a
particular key-value (even though many different users can verify
that one signature). If a second message needs to be signed, the
signer must choose a new private key and then distribute
corresponding new public keys. One open problem is thus whether
there exist reusable digital signature schemes, where either the
same copy of the public key can be used to verify many different
message-signature pairs securely, or where just the same key-values
can be used to verify many different message-signature pairs
securely (but a fresh copy of the public key is needed for each
verification). The latter notion of ``reusability'' is what we adopt
here.

In this paper, we consider an identification scheme, which, like a
digital signature scheme, is a type of authentication scheme.
Authentication schemes seek to ensure the \emph{integrity} of
information, rather than its privacy.  While digital signature
schemes ensure the integrity of origin of messages, identification
schemes ensure the integrity of origin of communication \emph{in
real time} \cite{MvOV96}.  Identification protocols are said to
ensure ``aliveness''---that the entity proving its identity is
active at the time the protocol is executed; we describe them in
more detail in the next section.

We prove that an identification scheme based on the one in Ref.
\cite{qphIM09}
 is secure against a computationally-unbounded adversary (only
restricted by finite cheating strategies), demonstrating for the
first time that unconditionally-secure and reusable public-key
authentication is possible in principle. We regard our result more
as a proof of concept than a (potentially) practical scheme.
 Still,
we are confident that an extension of the techniques used here may
lead to more efficient protocols.

We now proceed with a description of the protocol (Section
\ref{idscheme}) and the security proof (Section \ref{sec_security}).

\section{Identification Protocol}\label{idscheme}

In the following, Alice and Bob are always assumed to be honest
players and Eve is always assumed to be the adversary.  Suppose
Alice generates a private key and authentically distributes copies
of the corresponding public key to any potential users of the
scheme, including Bob.

Here is a description (adapted from Section 4.7.5.1 in Ref.
\cite{Gol01}) of how a secure public-key identification scheme
works.  When Alice wants to identify herself to Bob (i.e. prove that
it is she with whom he is communicating), she invokes the
identification protocol by first telling Bob that she is Alice, so
that Bob knows he should use the public key corresponding to Alice.
The ensuing protocol has the property that the \emph{prover} Alice
can convince the \emph{verifier} Bob (except, possibly, with
negligible probability) that she is indeed Alice, but an adversary
Eve cannot fool Bob (except with negligible probability) into
thinking that she is Alice, even after having listened in on the
protocol between Alice and Bob or having participated as a (devious)
verifier in the protocol with Alice several times. Public-key
identification schemes are used in smart-card systems (e.g., inside
an automated teller machine (ATM) for access to a bank account, or
beside a doorway for access to a building); the smart card
``proves'' its identity to the card reader.\footnote{Note that it is
not a user's personal identification number (PIN) that functions as
the prover's private key; the PIN only serves to authenticate the
user to the smart card (not the smart card to the card reader).}

 Note that no
identification protocol is secure against an attack where Eve
concurrently acts as a verifier with Alice and as a prover with Bob
(but note also that, in such a case, the ``aliveness'' property is
still guaranteed). Note also that, by our definition of
``reusable,'' an identification scheme is considered reusable if
Alice can prove her identity many times using the same key-values
but the verifier needs a fresh copy of the public key for each
instance of the protocol.

Note also that public-key identification can be trivially achieved
via a digital signature scheme (Alice signs a random message
presented by Bob), but we do not know of an unconditionally-secure
and reusable digital signature scheme.\footnote{Pseudo-signature
schemes, such as the one in Ref. \cite{CR90}, are
information-theoretically secure but assume broadcast channels.}
Similarly, public-key identification can be achieved with a
public-key encryption scheme (Bob sends an encrypted random
challenge to Alice, who returns it decrypted), but we do not know of
an unconditionally-secure and reusable public-key encryption scheme
(that uses pure-state public keys; though, see Ref. \cite{Got05} for
a promising candidate).

\subsection{Protocol specification}\label{sec_ProtSpec}

The identification protocol takes the form of a typical
``challenge-response'' interactive proof system, consisting of a
kernel (or subprotocol) that is repeated several times in order to
amplify the security, i.e., reduce the probability that Eve can
break the protocol. The following protocol is a simplification of
the original protocol from Ref. \cite{qphIM09} (but our security
proof applies to both protocols, with only minor adjustments).  We
assume all quantum channels are perfect.

\subsubsection*{Parameters}

\begin{itemize}
\item The \emph{security} parameter $s \in \mathbf{Z}^{+}$

\subitem $\diamond$ equals the number of kernel iterations.

\subitem $\diamond$ The probability that Eve can break the protocol
 is exponentially small in $s$.

\item The \emph{reusability} parameter $r \in \mathbf{Z}^{+}$

\subitem $\diamond$ equals the maximum number of copies of the
quantum public key in circulation and

\subitem $\diamond$ equals the maximum number of times the protocol
may be executed by Alice, before she needs to pick a new private
key.
\end{itemize}

\subsubsection*{Keys}

\begin{itemize}
\item The \emph{private key} is
\begin{eqnarray}\label{private_key}
(x_1,x_2,\ldots,x_s),
\end{eqnarray}

where Alice chooses each $x_j$, $j=1,2,\ldots,s$, independently and
uniformly randomly from $\{1,2,\ldots,2r+1\}$.

\subitem $\diamond$ The value $x_j$ is used only in the $j$th
kernel-iteration.

\item One copy of the \emph{public key} is an $s$-partite system in the state
\begin{eqnarray}\label{public_key}
\sys{\otimes_{j=1}^s \myket{\psi_{x_j}}},
\end{eqnarray}
where (omitting normalization factors)
\begin{eqnarray}\label{public_key}
\myket{\psi_{x_j}} := \ket{0} + e^{2\pi i x_j /(2r+1)}\ket{1}.
\end{eqnarray}

\subitem $\diamond$ Alice authentically distributes (e.g. via
trusted courier) at most $r$ copies of the public key.

\subitem $\diamond$ The $j$th subsystem of the public key (which
 is in the state $\sys{\myket{\psi_{x_j}}}$) is only used in the $j$th
 kernel-iteration.

\end{itemize}

\subsubsection*{Actions}

\begin{itemize}
\item The \emph{kernel} $\mathcal{K}(x)$ of the protocol is the following three steps,
where we use the shorthand
\begin{eqnarray}
\phi_x := 2 \pi x / (r+1),
\end{eqnarray}
and where we have dropped the subscript ``$j$'' from ``$x_j$'':

\subitem (1) Bob secretly chooses a uniformly random bit $b$ and
transforms the state of his authentic copy of $\myket{\psi_x}$ into
$\ket{0} + (-1)^b e^{i\phi_x}\ket{1}$.  Bob sends this qubit to
Alice.

\subitem (2) Alice performs the phase shift $\ket{1}\mapsto
e^{-i\phi_x}\ket{1}$ on the received qubit and then measures the
qubit in the basis $\{\ket{0} \pm \ket{1}\}$ (in order to determine
Bob's secret $b$ above). If Alice gets the outcome corresponding to
``+'', she sends 0 to Bob; otherwise, Alice sends 1.

\subitem (3) Bob receives Alice's bit as $b'$ and tests whether $b'$
equals $b$.

\item  When Alice wants to identify herself to Bob, they take the
following actions:

\subitem (\emph{i}) Alice checks that she has not yet engaged in the
protocol $r$ times before with the current value of the private key;
if she has, she aborts (and refreshes the private and public keys).

\subitem (\emph{ii}) Alice sends Bob her purported identity
(``Alice''), so that Bob may retrieve the public keys corresponding
to Alice.

\subitem (\emph{iii}) The kernel $\mathcal{K}(x)$ is repeated $s$
times, for $x=x_1, x_2,\ldots, x_s$.  Bob ``accepts'' if he found
that $b'$ equaled $b$ in all the kernel iterations; otherwise, Bob
``rejects''.

\end{itemize}

\subsection{Completeness of the protocol}

It is clear that the protocol is correct for honest players: Bob
always ``accepts'' when Alice is the prover.  In the Appendix
(``Section \ref{sec_security}''), we prove that the protocol is also
secure against any adversary (only restricted by finite cheating
strategies): given $r$ and $\epsilon
>0$, there exists a value of $s = s(r,\epsilon)$ such that Bob
``accepts'' with probability at most $\epsilon$ when Eve is the
prover.

\section{Security}\label{sec_security}

Let us clearly define what Eve is allowed to do in our attack model.
Eve can
\begin{itemize}
\item passively monitor Alice's and Bob's interactions (which means
that Eve can read the classical bits sent by Alice, and read the bit
that indicates whether Bob ``accepts'' or ``rejects''), and

\item participate as the verifier in one or more complete instances of the protocol, and

\item participate as the prover, impersonating Alice, in one or more complete instances of the protocol.
\end{itemize}
Eve is assumed not to be able to actively interfere with Alice's and
Bob's communications during the protocol, as this would allow Eve to
concurrently act as verifier with Alice and as prover with Bob (thus
trivially breaking any such scheme\footnote{For password-based
identification in a symmetric-key model, as in Ref. \cite{DFSS07},
where both Alice and Bob know something that Eve does not (i.e. the
password), one can define a nontrivial ``man-in-the-middle'' attack,
where Eve's goal is to learn the password in order to impersonate
Alice in a later instance of the protocol. However, in public-key
identification, Eve's goal of learning the private key may, without
loss of generality, be accomplished by participating as a dishonest
verifier and by obtaining copies of the public key, since Bob does
not perform any action that Eve cannot perform herself given a copy
of the public key.}).

Evidently, Eve's passive monitoring only gives her independent and
random bits (and the bit corresponding to ``accept''), thus giving
her no useful information (in that she may as well generate random
bits herself). So, we can ignore the effects of her passive
monitoring.

With regard to Eve acting as verifier, we will give Eve potentially
more power by assuming that Alice, instead of performing both the
phase shift and the measurement in Step 2 of the kernel
$\mathcal{K}(x)$, only performs the phase shift (Eve could perform
Alice's measurement herself, if she desired). Furthermore, we will
assume that the phase shift Alice performs is
\begin{eqnarray}
u_{\phi_x} = \begin{bmatrix} 1 & 0 \\ 0 & e^{i\phi_x}\end{bmatrix}.
\end{eqnarray}
Even though Alice actually performs the inverse phase shift $u_{-
\phi_x}$, note that the two phase shifts are equivalent in the sense
that $Z u_{\phi_x} Z$ equals $u_{-\phi_x}$ up to global phase, where
\begin{eqnarray}
Z = \begin{bmatrix} 1 & 0 \\ 0 & -1\end{bmatrix}.
\end{eqnarray}
Thus the protocol is unchanged had we assumed that Alice, instead of
performing $u_{-\phi_x}$ in Step 2 of the kernel $\mathcal{K}(x)$,
performs $Z u_{\phi_x} Z$.  Since Eve can perform $Z$ gates on her
qubit immediately before and after she gives it to Alice, our
assumption indeed gives Eve at least as much power to cheat.  Thus,
Eve can effectively extract up to $r$ black boxes for $u_{\phi_x}$
from Alice (recall Alice only participates in the protocol $r$ times
before refreshing her keys).

We will also give Eve potentially more power by giving her a black
box for $u_{\phi_x}$ in place of every copy of $\ket{\psi_{x}}$ that
she obtained legitimately. For each $x \in \{x_1,x_2\ldots,x_s\}$,
let $t$ be the total number of black boxes for $u_{\phi_x}$ that Eve
has in her possession; that is, for simplicity, and without loss of
generality, we assume she has the same number of black boxes
$u_{\phi_x}$ for each value of $x$. Note that $t \leq (2r-1)$, since
we always assume that at least one copy of the public key is left
for Bob, so that Eve can carry out the protocol with him.

Therefore, to prove security in our setting, it suffices to consider
attacks where Eve first uses her $st$ black boxes to create a
reference system in some $(\phi_{x_1}, \phi_{x_2}, \ldots,
\phi_{x_s})$-dependent state, denoted $\ket{\Psi_R(\phi_{x_1},
\phi_{x_2}, \ldots, \phi_{x_s})}$, and then she uses this system
while she participates as a prover, impersonating Alice, in one or
many instances of the protocol in order to try to cause Bob to
``accept''.  We use the following definition of ``security'':

\begin{definition}[Security]\label{def_Security}
An identification protocol (for honest prover Alice and honest
verifier Bob) is \emph{secure with error $\epsilon$} if the
probability that Bob ``accepts'' when any adversary Eve participates
in the protocol as a prover is less than $\epsilon$.
\end{definition}

\noindent The only assumption we make on Eve is that her cheating
strategy is finite in the sense that her quantum computations are
restricted to a finite-dimensional complex vector space; the
dimension itself, though, is unbounded.

We will assume that Eve has always extracted the $r$ black boxes for
$u_{\phi_x}$ from Alice (for all $x=x_1,x_2,\ldots,x_s$), and we
define $t'$ to be the number black boxes that Eve obtained
legitimately (via copies of the public key):
\begin{eqnarray}
t = r + t'.
\end{eqnarray}
Note that Eve can make at most $(r-t')$ attempts at fooling Bob,
i.e., causing Bob to ``accept''. Let $E(a,b)$ denote the event that
Eve fools Bob on her $a$th attempt using $b$ black boxes for
$u_{\phi_x}$ for all $x=x_1,x_2,\ldots,x_s$. Most of the argument,
beginning in Section \ref{sec_indiv_attacks}, is devoted to showing
that
\begin{eqnarray}\label{beginning}
\textrm{Pr}[E(1,t)] &\leq& (1-c/(t+2)^2)^s,
\end{eqnarray}
for some positive constant $c$ defined at the end of Section
\ref{sec_security}. In general, Eve learns something from one
attempt to the next; however, because Eve can simulate her
interaction with Bob at the cost of using one copy of
$\myket{\psi_x}$ per simulated iteration of $\mathcal{K}(x)$, we
have, for $\ell=2,3,\ldots,(r-t')$,
\begin{eqnarray}
&&\textrm{Pr}[E(\ell,t)]
%
\leq\textrm{Pr}[E(1,t+\ell-1)].
\end{eqnarray}
Given this, we use the union bound:
\begin{eqnarray}
&&\textrm{Pr}[\textrm{Eve fools Bob at least once, using $t$
black boxes for $u_{\phi_x}$, $\forall x$}]\hspace{4mm} \\
%
&\leq& \sum_{\ell=1}^{r-t'} \textrm{Pr}[E(\ell, t)]\\
%
%
&\leq& \sum_{\ell=1}^{r-t'} \textrm{Pr}[E(1,t+\ell-1]\\
 &\leq& \sum_{\ell=1}^{r-t'}
(1-c/(t+\ell+1)^2)^s\\
 &\leq& (r-t')(1-c/(2r+1)^2)^s,
\end{eqnarray}
since $t+\ell \leq 2r$. It follows that the probability that Eve can
fool Bob at least once, that is, break the protocol, is
\begin{eqnarray}\label{eqnPbreak}
P_{\textrm{\footnotesize break}} \leq r(1-c/(2r+1)^2)^s,
\end{eqnarray}
which, for fixed $r$, is exponentially small in $s$.  Note that this
bound is likely not tight, since it ultimately assumes that all of
Eve's attempts are equally as powerful.  In particular, this bound
assumes that Eve's state $\ket{\Psi_R(\phi_{x_1}, \phi_{x_2},
\ldots, \phi_{x_s})}$ does not degrade with use. A more detailed
analysis using results about degradation of quantum reference frames
 \cite{BRST06} may be possible.

From Eq. (\ref{eqnPbreak}) follows our main theorem (see Appendix
A.3 for the proof):
\begin{theorem}[Security of the protocol]\label{thm_Main} For any $\epsilon >0$ and any $r \in \mathbf{Z}^+$,
the identification protocol specified in Section \ref{sec_ProtSpec}
is secure with error $\epsilon$ according to Definition
\ref{def_Security} if
\begin{eqnarray}\label{eqn_scal}
s > (2r+1)^2\log(r/\epsilon)/c,
\end{eqnarray}
for some positive constant $c$.
\end{theorem}
\noindent The theorem shows how the efficiency of the protocol
scales with its reusability: it suffices to have
\begin{eqnarray}
s \in O(r^2 \log(r/\epsilon)).
\end{eqnarray}

The remainder of the paper establishes the bound in Line
(\ref{beginning}).

\subsection{Sufficiency of individual attacks}\label{sec_indiv_attacks}

At each iteration, we may assume Eve performs some measurement, in
order to get an answer to send back to Bob.  Generally, Eve can
mount a coherent attack, whereby her actions during iteration $j$
may involve systems that she used or will use in previous or future
iterations as well as systems created using black boxes for
$u_{\phi_{x_k}}$ for any $k$---not just for $k = j$.  Since each
$x_j$ is \emph{independently} selected from the set $\{1, 2, \ldots,
2r+1\}$, intuition suggests that Eve's measurement at iteration $j$
may be assumed to be independent of her measurement at any other
iteration and in particular does not need to involve any black boxes
other than ones for $u_{\phi_{x_j}}$. In other words, it seems
plausible that the optimal strategy for Eve can consist of the
``product'' of identical optimal strategies for each iteration
individually.  This intuition can indeed be shown to be correct by
combining a technique from Ref. \cite{Gut09}, for expressing the
maximum output probability in a multiple-round quantum interactive
prototol as a semidefinite program, with a result in Ref.
\cite{MS07}, which implies that the semidefinite program satisfies
the product rule that we need; see Appendix A.1 for a proof.

The remainder of Section \ref{sec_security} establishes the
following proposition:
\begin{prop}\label{prop_MaxGuessProbIndiv} The probability that Eve guesses correctly in any particular
iteration $j$, using $t$ black boxes for $u_{\phi_{x_j}}$, is at
most $(1-c/(t+2)^2)$ for some positive constant $c$.
\end{prop}
\noindent Assuming Proposition \ref{prop_MaxGuessProbIndiv}, the
result proved in Appendix A.1 implies that the probability of Eve's
guessing correctly in all $s$ iterations, using $t$ black boxes for
$u_{\phi_x}$, for $x = x_1,x_2,\ldots,x_s$, is at most
$(1-c/(t+2)^2)^s$, establishing the bound in Line (\ref{beginning}).

\subsection{Equivalence of discrete and continuous private phases}

To help us prove Proposition \ref{prop_MaxGuessProbIndiv}, we now
show that, from Bob's and Eve's points of view, Alice's choosing the
private phase angle $\phi_x$ from the discrete set $\{2 \pi
x/(2r+1): x=1,2,\ldots,2r+1\}$ is equivalent to her choosing the
phase angle from the continuous interval $[0,2\pi)$. We have argued
that the only information that Eve or Bob---or anyone but
Alice---has about $\phi_x$ may be assumed to come from a number of
black boxes for $u_{\phi_x}$ that can be no greater than $2r$ (there
are $r$ legitimate copies of the public key, and one can extract $r$
more black boxes from Alice); let this number be $d$, where $1 \leq
d \leq 2r$.

In order to access the information from the black boxes, they must,
in general, be used in a quantum circuit in order to create some
state. Using the $d$ black boxes, the most general (purified) state
that can be made is without loss of generality of the form
\begin{eqnarray}
\ket{\psi(\phi_x)} = \sum_{k=0}^{N-1} \left( \sum_{j=0}^{d}
\beta_{j,k} e^{i j \phi_x} \right) \myket{a_k},
\end{eqnarray}
where $\{\myket{a_k}: k = 0,1,..., N-1\}$ is an orthonormal basis of
arbitrary but finite size (the assumption of finite $N$ comes from
our restricting Eve to using only finite cheating strategies). In
general, the numbers $N$ and $\beta_{j,k}$ may depend on $d$. Here
we have followed Ref. \cite{DAEMM07} by noting that each amplitude
is a polynomial in $e^{i \phi_x}$ of degree at most $d$; this fact
follows from an inductive proof just as in Ref. \cite{BBCMW98},
where the polynomial method is applied to an oracle revealing one of
many Boolean variables.

Averaging over Alice's random choices of $x$, one would describe the
previous state by the density operator
\begin{eqnarray}
\frac{1}{2r+1}\sum_{x = 1}^{2r+1}
\ketbra{\psi(\phi_x)}{\psi(\phi_x)},
\end{eqnarray}
since $x$ is chosen uniformly randomly from $\{1,2,\ldots,2r+1\}$.
Had $\phi_x$ been chosen uniformly from $\{2 \pi x/(2r+1): x\in
[0,2r+1)\} = [0,2\pi)$, one would describe the state by
\begin{eqnarray}
\int_0^{2 \pi} \frac{d\phi}{2\pi} \ketbra{\psi(\phi)}{\psi(\phi)}.
\end{eqnarray}
It is straightforward to show\footnote{This requires the following
two facts: (1) for any integer $a$,
\begin{eqnarray}
\frac{1}{2\pi}\int_{0}^{2\pi}  e^{i a \theta} d\theta= \left\{
\begin{array}{ll}
                0 & \mbox{ if $a\neq 0$ }, \\
                1 & \mbox{ otherwise     };
                \end{array}
                \right.
\end{eqnarray}
and (2) for any integer $p\geq 2$ and integer $a$:
\begin{eqnarray}
\frac{1}{p}\sum_{k=1}^p e^{2 \pi i a k/p} =\left\{
\begin{array}{ll}
                0 & \mbox{ if $a$ is not a multiple of $p$}, \\
                1 & \mbox{ otherwise     },
                \end{array}
                \right.
\end{eqnarray}
where the second fact is applied at $p=2r+1$.}~that the above two
density operators are both equal to
\begin{eqnarray}
\sum_{k, k' = 0}^{N-1} \sum_{j = 0}^{d} \beta_{j,k} \beta_{j, k'}^*
\ketbra{a_{k}}{a_{k'}}.
\end{eqnarray}
Thus, without loss of generality, we may drop the subscript ``$x$''
on ``$\phi_x$'', write ``$\phi$'' for Alice's private phase angle,
and assume she did (somehow) choose $\phi$ uniformly randomly from
$[0,2\pi)$.\footnote{One way to interpret this result is that even
if Alice encodes infinitely many bits into $\phi$, it is no better
than if she encoded $\lceil \log_2(2r+1) \rceil$ bits. Note that if
Eve performs an optimal phase estimation \cite{vDdAEMM07} in order
to learn $\phi$ and then cheat Bob, she can only learn at most
$\lfloor \log_2(2r-1)\rfloor$ bits of $\phi$ (here, we assume Eve
has $2r-1$ copies of the public key, having left Bob one copy),
whereas Alice actually encoded $\lceil \log_2(2r+1) \rceil$ bits
into $\phi$.}
 ~We are now ready to prove Proposition \ref{prop_MaxGuessProbIndiv}.

\subsection{Bound on relative phase shift estimation}

Eve's task of cheating in one iteration of the kernel may be phrased
as follows.  Eve is to decide the difference between the relative
phases encoded in two subsystems $R$ and $S$, where $S$ is a given
one-qubit system and $R$ is under her control. The given subsystem
$S$ is in the state
\begin{eqnarray}
\ket{\psi_S(\phi,\theta)} = \myket{0} + e^{i(\phi +
\theta)}\myket{1},
\end{eqnarray}
where $\theta$ is unknown and uniformly random in $\{0, \pi\}$, and
$\phi$ is unknown and uniformly random in
$[0,2\pi]$.  Eve can make the state $\ket{\psi_R(\phi)}$ of
subsystem $R$ by using arbitrary operations interleaved with at most
$t$ black boxes for the one-qubit gate $u_\phi$. Note that the
problem is nontrivial because $\phi$ is unknown and uniformly random
and the qubit $S$ is given to Eve \emph{after} she has used all her
black boxes. We seek the optimal success probability for Eve to
guess $\theta$ correctly.

Eve's estimation problem can be treated within the framework of
quantum estimation of group transformations \cite{CAS05}.  As such,
we regard her problem as finding the optimal measurement
(probability) to correctly distinguish the states in the two-element
orbit
\begin{eqnarray}
\{V_\theta \rho V_\theta^\dagger: \theta \in \{0,\pi\}\},
\end{eqnarray}
where $V_\theta = I_R \otimes (\ketbra{0}{0}+e^{i \theta}
\ketbra{1}{1})$ and
\begin{eqnarray}
\rho = \int \frac{d\phi}{2\pi}
\myketbra{\psi_R(\phi)}{\psi_R(\phi)}\otimes
\myketbra{\psi_S(\phi,0)}{\psi_S(\phi,0)}.
\end{eqnarray}
The probabilities of her estimation procedure can be assumed to be
generated by a POVM $\{E_0, E_\pi\}$.  In general, it is known how
to solve for the POVM that performs optimally on average when the
unitarily-generated orbit consists of pure states, but not when the
orbit is generated from a mixed state ($\rho$, in our case). Thus,
we now effectively reduce the problem to several instances of an
estimation problem where the orbit is pure.

Indeed, suppose that $\ket{\psi_R(\phi)}$ were a state on $q$ qubits
that satisfied the property
\begin{eqnarray}\label{eqn_cov_sig_states}
\myketbra{\psi_R(\phi)}{\psi_R(\phi)} = (u_\phi)^{\otimes
q}\myketbra{\psi_R(0)}{\psi_R(0)}(u_\phi^\dagger)^{\otimes q}
\end{eqnarray}
for all $\phi \in [0,2\pi]$.  Then, letting $U_\phi \equiv
(u_\phi)^{\otimes (q+1)}$ and $\ket{\psi_{RS}(\phi,\theta)}\equiv
\myket{\psi_R(\phi)}\myket{\psi_S(\phi,\theta)}$, we would have that
\begin{eqnarray}
\rho &=& \int \frac{d\phi}{2\pi} U_\phi \myket{\psi_{RS}(0,0)}
 \mybra{\psi_{RS}(0,0)}U_\phi^\dagger\\
&=& \sum_w P_w  \myket{\psi_{RS}(0,0)}
 \mybra{\psi_{RS}(0,0)} P_w\label{expressionForRho}\\
 &=& \sum_w P_w  \rho P_w,
\end{eqnarray}
where $P_w$ is the projection onto the subspace of Hamming weight $w
= 0,1,\dots, q+1$, and we used the formulas $U_\phi = \sum_w P_w
e^{i w \phi}$ and $\delta_{w,0}=\int (d\phi/2\pi)e^{i w \phi}$. In
other words, the state $\rho$ would be block diagonal with respect
to the direct-sum decomposition of the total state space of $R$ into
subspaces of constant Hamming weight $w$. Then we would have that
the probability that Eve guesses $\theta = \theta'$ given that
$\theta = \theta''$ is
\begin{eqnarray}
\textrm{Pr}[\textrm{Eve guesses }\theta=\theta' | \theta=\theta'']
&=& \tr\left[E_{\theta'} \left( V_{\theta''} \rho V_{\theta''}^\dagger \right)\right]\\
&=& \tr\left[E_{\theta'} V_{\theta''} \sum_w P_w \rho P_w V_{\theta''}^\dagger \right]\\
%
%
&=& \tr \left[ \left( \bigoplus_{w} E_{w, \theta'} \right)  \left(
V_{\theta''} \rho V_{\theta''}^\dagger \right)
 \right],
%
%
\end{eqnarray}
where $E_{w,\theta'} \equiv P_w E_{\theta'} P_w$, and we used
cyclicity of trace and the fact that $V_\theta$ and $P_w$ commute.
Thus, the elements of Eve's POVM $\{E_0, E_\pi\}$ would without loss
of generality have the same block diagonal structure as $\rho$.  In
principle, this would allow Eve to measure first (just) the Hamming
weight of $\rho$ in order to find $w$, and then deal with the group
transformation estimation problem with respect to the pure orbit
\begin{eqnarray}
\mathcal{O}_w \equiv \{V_\theta \ket{\Psi_w} : \theta \in
\{0,\pi\}\},
\end{eqnarray}
where $\ket{\Psi_w}$ is the state such that $\ket{\Psi_w}\propto P_w
\myket{\psi_{RS}(0,0)}$; we note that $\ket{\Psi_w}$ is independent
of $\phi$ (and $\theta$). The following lemma shows that, without
loss of generality, we may assume that the situation just described
is indeed the case:

\begin{lemma}\label{lem_Main}
Without loss of generality, Eve's state $\myket{\psi_R(\phi)}$,
which she prepares with at most $t$ black boxes for $u_\phi$, may be
assumed to be on $q = (2t+1)$ qubits and satisfy
\begin{eqnarray}
\myketbra{\psi_R(\phi)}{\psi_R(\phi)} = (u_\phi)^{\otimes
q}\myketbra{\psi_R(0)}{\psi_R(0)}(u_\phi^\dagger)^{\otimes q}
\end{eqnarray}
 for all
$\phi \in [0,2\pi]$ .
\end{lemma}

\begin{proof}
As noted in the previous section, using the $t$ black boxes, the
most general (purified) state of $R$ that Eve can make is without
loss of generality
\begin{eqnarray}\label{eqn_MostGenState}
\sum_{k=0}^{N-1} \left( \sum_{j=0}^{t} \beta_{j,k} e^{i j \phi}
\right) \myket{a_k}_R,
\end{eqnarray}
where, again, $N$ is a priori unknown but finite (we use subscripts
on the kets in this proof to indicate the physical systems).  Note
that we can rewrite the state in Eq. (\ref{eqn_MostGenState}) by
changing the order of the summations as
\begin{eqnarray}\label{eqn_MostGenStateReWritten}
\sum_{j=0}^{t} \beta_j e^{i j \phi} \myket{\tilde{g}_j}_R,
\end{eqnarray}
where we have defined the numbers $\beta_j$ and the
not-necessarily-orthogonal set of unit vectors $\{
\myket{\tilde{g}_j} : j=0,1,...,t\}$ such that
\begin{eqnarray}
\beta_j \myket{\tilde{g}_j}_R = \sum_{k=0}^{N-1} \beta_{j,k}
\myket{a_k}_R .
\end{eqnarray}
Using the Gram-Schmidt orthonormalization procedure on
$\{\myket{\tilde{g}_j}\}_j$ to get the orthonormal set
$\{\myket{g_j}\}_j$, we can write
\begin{eqnarray}
\myket{\tilde{g}_j}_R  = \sum_{h=0}^t \gamma_{j,h} \myket{g_h}_R.
\end{eqnarray}
Introduce a new system $R'$ consisting entirely of qubits and define
$U$ to be any unitary map acting on $R \otimes R'$ that takes
$\myket{0}_R \myket{c_h}_{R'} \mapsto \myket{g_h}_R \myket{0}_{R'}$,
where $\{ \myket{c_h}_{R'} \}_{h=0,1,\ldots,t}$ is an orthonormal
set of size $t+1$ with elements that are computational basis states
whose labels have constant Hamming weight; note that $R'$ needs only
$O(\log(t+1))$ qubits whereas $R$ is of unknown (but finite) size
(however, following this proof, we will construct $R'$ using $t+1$
qubits, as this makes things simpler). We first claim that, without
loss of generality,
\begin{eqnarray}\label{eqn_WLOG_RFstate}
\ket{\psi_R(\phi)}=\sum_{j,h} \beta_j  \gamma_{j,h} e^{i
j\phi}\myket{S^t_j}_{A} \myket{c_h}_{R'},
\end{eqnarray}
where $A$ is a $t$-qubit ancilla, and $\myket{S^t_j}_A$ is the
symmetric state of weight $j$.~
%
%
%
 To see this, note that
Eve's optimal measurement can include the following pre-processing
operations (in sequence), so that she recovers the most general
state in Eq. (\ref{eqn_MostGenState}) (and Eq.
(\ref{eqn_MostGenStateReWritten})) on $R$ but for a different random
value of $\phi$:
\begin{itemize}

\item add an ancillary register $R$ in state $\myket{0}_{R}$ in
between the two registers $A$ and $R'$ and perform $U$ on $R\otimes
R'$ to get (after throwing out system $R'$)
\begin{eqnarray}
\sum_j \beta_j \sum_h \gamma_{j,h} e^{i j\phi}\myket{S^t_j}_{A}
\myket{g_h}_R
 = \sum_j \beta_j e^{i j\phi}\myket{S^t_j}_{A} \myket{\tilde{g}_j}_R
\end{eqnarray}

\item on $A$, do the $(t+1)$-dimensional inverse quantum Fourier transform in the
symmetric basis on $A$, i.e. mapping
\begin{eqnarray}
\myket{S^t_j}_A \mapsto \frac{1}{\sqrt{t+1}} \sum_y e^{- i 2 \pi y
j/(t+1)}\ket{S^t_y}_A,
\end{eqnarray}
to get
\begin{eqnarray}
\sum_j \sum_y \beta_j e^{i j(\phi - 2\pi y/(t+1))}\myket{S^t_y}_{A}
\myket{\tilde{g}_j}_R
\end{eqnarray}
and measure the Hamming weight of $A$ to get result $y_0$, which
leaves the state (after throwing out system $A$)
\begin{eqnarray}
\sum_j \beta_j e^{i j(\phi - 2\pi y_0/(t+1))} \myket{\tilde{g}_j}_R
\end{eqnarray}

\item correct the relative phase on qubit $S$ by $2 \pi y_0/(t+1)$.

\end{itemize}
Doing these operations does not change the estimation problem, since
$\phi$ is uniformly random anyway; these operations just change the
unknown $\phi$ to $\phi' = \phi - 2 \pi y_0 /(t+1)$.

Finally, note that Eq. (\ref{eqn_WLOG_RFstate}) implies that
$\ket{\psi_R(\phi)}$ can be made from $\ket{\psi_R(0)}$ with at most
$t$ black boxes for $u_\phi$, by applying $(u_\phi)^{\otimes t}$ on
the $t$ qubits of system $A$, and note that $\ket{\psi_R(\phi)}$
satisfies Eq. (\ref{eqn_cov_sig_states}), since the states
$\myket{c_h}$ are of constant Hamming weight.
\end{proof}

\begin{remark}[Quantum Fourier transform as analytical tool] Note that Eve's optimal strategy is not necessarily
to measure $R$ to get an estimate $\phi'$ of $\phi$ first, then
apply $u_{-\phi'}$ on $S$, and then measure $S$ to estimate
$\theta$. However, the operation that is optimal for estimating
$\phi$ (see Ref. \cite{DAEMM07}), i.e. the inverse quantum Fourier
transform applied above, is still useful as an analytical tool in
order to derive (a convenient form of) an optimal state for her
estimation of $\theta$.
\end{remark}

Thus, by Lemma \ref{lem_Main}, we assume Eq.
(\ref{eqn_WLOG_RFstate}) holds, which allows us to derive the
following proposition.  For convenience, we define
\begin{eqnarray}
\alpha_{j,h}\equiv \beta_j \gamma_{j,h}.
\end{eqnarray}
\begin{prop}\label{prop_DerivationOfPOVM}
The elements of the POVM $\{E_0,E_\pi\}$ are without loss of
generality defined as
\begin{eqnarray}
E_0 &=& \myket{\Xi_{0}}\ket{0}\mybra{\Xi_0}\bra{0}
+\sum_{w=2}^{t+1}\ketbra{w,+}{w,+} \\
E_\pi &=& \sum_{w=2}^{t+1}\ketbra{w,-}{w,-}
+\myket{\Xi_t}\ket{1}\mybra{\Xi_t}\bra{1},
%
\end{eqnarray}
where
%
\begin{eqnarray}
\myket{w,\pm} \equiv \frac{1}{\sqrt{2}} (\ket{\Xi_{w-1}}\ket{0} \pm
\ket{\Xi_{w-2}}\ket{1}),
\end{eqnarray}
 and $\myket{\Xi_{w-1}}$ and $\myket{\Xi_{w-2}}$ are
states such that, for $j=0,1,\ldots,t$,
%
\begin{eqnarray}\label{def_Xis}
\myket{\Xi_{j}} \propto \sum_h
\frac{\alpha_{j,h}}{\sqrt{2}}\myket{S^t_j}\myket{c_h}.
\end{eqnarray}
\end{prop}
\noindent The proof of Proposition \ref{prop_DerivationOfPOVM} is
similar to the argument given in Ref. \cite{BRST06} and is given in
Appendix A.2. The total success probability of Eve's strategy can
now be computed as
\begin{eqnarray}
&&\sum_{\theta' \in \{0,\pi\}}\textrm{Pr}[\textrm{Eve guesses
}\theta=\theta' | \theta=\theta']\textrm{Pr}[\theta=\theta']\\
&=& \frac{1}{2} \sum_{\theta' \in \{0,\pi\}} \tr (E_{\theta'}
V_{\theta'} \rho  V_{\theta'}^\dagger) \\
&=& \frac{1}{2} \sum_{\theta' \in \{0,\pi\}} \tr (E_{\theta'}
V_{\theta'} \myket{\psi_{RS}(0,0)}
 \mybra{\psi_{RS}(0,0)}  V_{\theta'}^\dagger) \\
&=&\label{line_tomaximize}\frac{1}{2} +
\frac{1}{4}\mybra{\psi_R(0)}M_t \myket{\psi_R(0)},
\end{eqnarray}
where
\begin{eqnarray}
M_t \equiv \sum_{j=0}^{t-1} \myketbra{\Xi_{j+1}}{\Xi_{j}}
+\myketbra{\Xi_{j}}{\Xi_{j+1}}.
\end{eqnarray}

As a last task, we now seek the value of $\myket{\psi_R(0)}$---i.e.
the values of $\alpha_{j,h}$---such that $\mybra{\psi_R(0)}M_t
\myket{\psi_R(0)}$ is maximal.  The proof of the following
proposition is in Appendix A.4:
\begin{prop}\label{prop_Maximization}
The state $\myket{\psi_R(0)} \propto \sum_{j=0}^t  \sin
\left[\frac{(j+1)\pi}{t+2}\right] \myket{\Xi_j}$ achieves the
maximum value in Eq. (\ref{line_tomaximize}).
\end{prop}

\noindent Thus (as in Ref. \cite{BRST06}---see Appendix A.4), we get
a maximal success probability of
\begin{eqnarray}
&&\frac{1}{2} + \frac{1}{2}\cos(\pi/(t+2)) \\
&\leq& \frac{1}{2} + \frac{1}{2}\left( 1 -\frac{(\pi/(t+2))^2}{2!} +
\frac{(\pi/(t+2))^4}{4!} \right)\\
&=& 1 - \frac{\pi^2}{4}\frac{1}{(t+2)^2} + \frac{\pi^4}{48}\frac{1}{(t+2)^4}\\
&\leq& 1 - \left(\frac{\pi^2}{4} -\frac{\pi^4}{48}\right)\frac{1}{(t+2)^2}\\
&=& 1 - c/(t+2)^2,
\end{eqnarray}
for the constant $c = (\pi^2/4 - \pi^4/48) \doteq 0.438$ and all
$t\geq 1$.  This completes the proof of Proposition
\ref{prop_MaxGuessProbIndiv} and thus the proof of Theorem
\ref{thm_Main}.

\bibliographystyle{unsrt}




\section*{Appendices}

\subsection*{A.1~~~Proof of sufficiency of individual attacks}

Consider the following non-cryptographic, $(t+1)$-round interactive
protocol (or game) between Evelyn and Bobby (neither of whom is
considered adversarial, hence we distinguish these two players from
Eve and Bob), denoted $\mathcal{L} = \mathcal{L}(\Phi)$, where
\begin{eqnarray}
\Phi = (\Phi_1,\Phi_2,\ldots,\Phi_{t+1})
\end{eqnarray}
and the $\Phi_i$ are quantum operations (super-operators) that
specify Evelyn's actions in the game (the quantities $r$ and $t$ are
as defined previously):

\begin{itemize}
\item $(1')$ Bobby chooses a uniformly random $x\in
\{1,2,\ldots,2r+1\}$ and sends a qubit in the state $\ket{0}$ to
Evelyn (who can ignore this qubit---it carries no significant
information).

\item $(2')$  For $i = 1,2,\ldots, t$ $\{$

\subitem \hspace{5mm}$\diamond$ Evelyn performs the quantum
operation $\Phi_i$ on her system, and then sends one qubit to Bobby.

\subitem \hspace{5mm}$\diamond$ Bobby performs the unitary gate
$u_{\phi_x}$ on the qubit received from Evelyn and sends it back to
Evelyn.$\}$

\item $(3')$ Bobby chooses a uniformly random $b \in \{0,1\}$ and
sends a qubit in the state $\ket{0} + (-1)^b e^{i \phi_x}\ket{1}$ to
Evelyn.

\item $(4')$ Evelyn performs the quantum operation $\Phi_{t+1}$ on her
system, and then sends one qubit to Bobby.

\item $(5')$ Bobby measures the received qubit in the computational basis $\{\ket{0},
\ket{1}\}$, getting outcome 0 or 1 (corresponding to $\ket{0}$ and
$\ket{1}$ respectively); he tests whether this outcome equals $b$.
\end{itemize}

The following proposition is straightforward to prove:

\begin{prop} The probability that Eve, using $t$ black boxes $u_{\phi_{x_j}}$, causes Bob's equality test to
pass in a particular iteration $j$ of the protocol in Section
\ref{sec_ProtSpec} is at most
\begin{eqnarray}
\alpha := \max_{\Phi} \textrm{\emph{Pr}}[\textrm{\emph{Bobby's
equality test passes in }}\mathcal{L}(\Phi)],
\end{eqnarray}
where $\Phi$ ranges over all $(t+1)$-tuples of admissible quantum
operations that Evelyn can apply in the game $\mathcal{L}$.
\end{prop}

Now consider the parallel $s$-fold repetition of $\mathcal{L}$,
which we denote $\mathcal{L}^{\| s} = \mathcal{L}^{\| s} (\Phi')$,
where now $\Phi'$ denotes Evelyn's quantum operation in
$\mathcal{L}^{\| s}$.  The following proposition is also
straightforward to prove:

\begin{prop}  The probability that Eve fools Bob on the first attempt
using $t$ black boxes per $x$-value in the protocol in Section
\ref{sec_ProtSpec} is at most
\begin{eqnarray}
\alpha':= \max_{\Phi'} \textrm{\emph{Pr}}[\textrm{\emph{all of
Bobby's equality tests pass in }}\mathcal{L}^{\| s}(\Phi')],
\end{eqnarray}
where $\Phi'$ ranges over all $(t+1)$-tuples of admissible quantum
operations that Evelyn can apply in the game $\mathcal{L}^{\| s}$.
\end{prop}

Therefore, in order to prove that it is sufficient to consider
individual (as opposed to coherent) attacks by Eve, it suffices to
show that $\alpha' = \alpha^s$.

In Ref. \cite{Gut09}, the above game is viewed as an interaction
between a $(t+1)$-round \emph{(non-measuring) strategy} and a
 \emph{(compatible) measuring co-strategy}; Evelyn's operations $\Phi$
form the non-measuring strategy and Bobby's actions form the
measuring co-strategy (technically, Steps $(1')$, $(3')$, and $(4')$
would have to be slightly modified in order to fit the co-strategy
formalism: in Steps $(1')$ and $(3')$, Bobby should make his random
choices in superposition and use the quantum registers storing these
choices as a control register whenever requiring these random values
subsequently; in Step $(4')$, Bobby should only make one final
measurement whose outcome indicates whether the equality test
passes; we assume that these modifications have been made).

For all $i$, let $\X_i$ and $\Y_i$ be the input and output spaces,
respectively, of Evelyn's quantum operation $\Phi_i$ in
$\mathcal{L}$, i.e. $\Phi_i: \lin{\X_i} \rightarrow \lin{\Y_i}$,
where $\lin{\X_i}$ is the space of all linear operators from the
complex Euclidean space $\X_i$ to itself (and likewise for
$\lin{\Y_i}$). Let $\pos{\Y \otimes \X}$ denote the set of all
positive semidefinite operators in $\lin{\Y \otimes \X}$, where $\Y
= \Y_1 \otimes \Y_2 \otimes \cdots \otimes \Y_{t+1}$ (and similarly
for $\X$). For any Euclidean space $\Z$, let $\I_{\Z}$ denote the
identity operator $\Z$.

Ref. \cite{Gut09} shows that Evelyn's strategy can be equivalently
expressed by a single positive semidefinite operator in $\pos{\Y
\otimes \X}$ while Bobby's measuring co-strategy can be expressed by
the collection $\{B_0, B_1\}$ of two positive semidefinite operators
in $\pos{\Y \otimes \X}$, where, without loss of generality, we
assume that $B_0$ corresponds to the measurement outcome indicating
that Bobby's test for equality in Step $(5')$ passes.  We briefly
note that these positive semidefinite operators are the
Choi-Jamio{\l}kowski representations of quantum operations
corresponding to the players' actions.  A more general version of
the following theorem is proved in Ref. \cite{Gut09}:

\begin{theorem}[Interaction output probabilities
\cite{Gut09}]\label{thm_Gut_OutputProbs} For any non-measuring
strategy $X\in \pos{\Y \otimes \X}$ of Evelyn, the probability that
Bobby's equality test passes is $\tr (B_0^\dagger X)$.
\end{theorem}

\noindent Using Theorem \ref{thm_Gut_OutputProbs}, it is shown, in
the proof of Theorem 3.3 of Ref. \cite{Gut09}, that the maximal
probability with which Bobby's measuring co-strategy can be forced
to output the outcome corresponding to $B_0$ by some (compatible)
strategy of Evelyn's can be expressed as a semidefinite
(optimization) program (see Ref. \cite{Wat08} for a relevant review
of semidefinite programming).
 Thus $\alpha$ and
$\alpha'$ can be expressed, respectively, as solutions to the
following semidefinite programs $\pi_{\alpha}$ and $\pi_{\alpha'}$:
%
%
\begin{eqnarray}\nonumber
\underline{\pi_\alpha} \hspace{15mm}&\hspace{5mm}&
\hspace{15mm}\underline{\pi_{\alpha'}}\\\nonumber
    \text{maximize:} \hspace{2mm}\tr(B_0^\dagger X) \hspace{8mm}& \hspace{5mm}  &  \text{maximize:} \hspace{2mm}\tr((B_0^{\otimes s})^\dagger
    X)\\\nonumber
\text{subject to:} \hspace{2mm}
    \tr_{\Y} (X) = \I_{\X}, & \hspace{5mm} &\text{subject to:}
    \hspace{2mm}
    \tr_{\Y'} (X) = \I_{\X'},\\\nonumber
X\in \pos{\Y\otimes \X} &\hspace{5mm}& \hspace{14mm}X\in \pos{\Y'
\otimes \X'},
\end{eqnarray}
where, for all $i$, $\X_i' = \X_i^{\otimes s}$ and $\X' = \X_1'
\otimes \X_2' \otimes \cdots \otimes \X_{t+1}'$ (and similarly for
$\Y_i'$ and $\Y'$).  We note that the first constraint in each
semidefinite program above codifies the property of
trace-preservation for the quantum operation corresponding to $X$,
while the second constraint codifies the property of complete
positivity (see Ref. \cite{Wat08} for details).  Furthermore, it is
shown in Ref. \cite{Gut09} that such semidefinite programs (arising
from interactions between strategies and compatible co-strategies)
satisfy the condition of strong duality, which means that the
solution to each semidefinite program above coincides with that of
its dual.

In Ref. \cite{MS07}, the following theorem is proven:

\begin{theorem}[Condition for product rule for semidefinite programs \cite{MS07}]\label{thm_MittalSzegedy}
Suppose that the following two semidefinite programs $\pi_1$ and
$\pi_2$ satisfy strong duality:
%
\begin{eqnarray}\nonumber
\underline{\pi_1} \hspace{15mm}&\hspace{5mm}&
\hspace{15mm}\underline{\pi_2}\\\nonumber
    \text{maximize:} \hspace{2mm}\textrm{\emph{Tr}}(J_1^\dagger W) \hspace{6mm}& \hspace{5mm}  &  \text{maximize:} \hspace{2mm}\textrm{\emph{Tr}}(J_2^\dagger W)
    \\\nonumber
\text{subject to:} \hspace{2mm}
    \Psi_1 (W) = C_1, & \hspace{5mm} &\text{subject to:}
    \hspace{2mm}
    \Psi_2 (W) = C_2,\\\nonumber
W\in\pos{\W_1} &\hspace{5mm}& \hspace{14mm}W\in\pos{\W_2},
\end{eqnarray}
where $\Psi_1: \lin{\W_1} \rightarrow \lin{\Z_1}$ and $\Psi_2:
\lin{\W_2} \rightarrow \lin{\Z_2}$, for complex Euclidean spaces
$\W_1, \Z_1, \W_2, \Z_2$, and $J_1 \in \lin{\W_1}$ and
$J_2\in\lin{\W_2}$ are Hermitian.  Let $\alpha(\pi_1)$ and
$\alpha(\pi_2)$ denote the semidefinite programs' solutions.  If
$J_1$ and $J_2$ are positive semidefinite, then the solution to the
following semidefinite program, denoted $\pi_1 \otimes \pi_2$, is
$\alpha(\pi_1 \otimes \pi_2) = \alpha(\pi_1) \alpha(\pi_2)$:
\begin{center}
  \begin{minipage}{3in}
    \centerline{\underline{$\pi_1 \otimes \pi_2$}}\vspace{-7mm}
    \begin{eqnarray*}
    \text{maximize:} && \textrm{\emph{Tr}}( (J_1\otimes J_2)^\dagger W)\\
    \text{subject to:} &&
    \Psi_{1} \otimes \Psi_2 (W) = C_1 \otimes C_2,\\
    && W\in\pos{\W_1 \otimes \W_2}.
    \end{eqnarray*}
  \end{minipage}
\end{center}
\end{theorem}
Since $B_0$ is positive semidefinite and $\pi_{\alpha'} =
\pi_\alpha^{\otimes s}$ (using the associativity of $\otimes$),
Theorem \ref{thm_MittalSzegedy} can be applied $(s-1)$ times in
order to prove that $\alpha' = \alpha^s$ as required.  See Ref.
\cite{Gut09} for a similar approach, based on ideas in Ref.
\cite{CSUU06}.  The idea of expressing the acceptance probability of
a quantum interactive proof system as a semidefinite program first
appeared in Ref. \cite{KW00}.

Note that this argument, combined with the arguments in the main
body of the paper, shows that both the serial and parallel versions
of our identification protocol are secure.

\subsection*{A.2~~~Proof of Proposition \ref{prop_DerivationOfPOVM}}

Two facts hold without loss of generality:
\begin{itemize}

\item the POVMs $\{E_{w,0}, E_{w,\pi}\}$, for
all $w$, may be assumed to be covariant, i.e. $E_{w,\pi} = V_\pi
E_{w,0} V_\pi^\dagger$ (to see this, note that any
not-necessarily-covariant POVM $\{F_{w,0}, F_{w,\pi}\}$ gives the
same average probability of successfully guessing $\theta$, given
$w$, as the covariant POVM $\{E_{w,0}, E_{w,\pi}\}$ defined by
$E_{w,0} = (F_{w,0} + V_\pi^\dagger F_{w, \pi} V_\pi)/2$);

\item each $E_{w,0}$ has support only on $\sp( \mathcal{O}_w)$ and thus $E_{w,0} + E_{w,\pi} = I_{\sp( \mathcal{O}_w)}$, where $I_{\sp( \mathcal{O}_w)}$ is the identity operator on $\sp( \mathcal{O}_w)$.

\end{itemize}
\noindent To compute a basis of $\sp( \mathcal{O}_w)$, we now
further define the system $R'$ in the proof of Lemma \ref{lem_Main}
to consist of exactly $t+1$ qubits and the states $\ket{c_h}$,
$h=0,1,\ldots,t$, to be all those computational basis states whose
labels have Hamming weight 1 (thus $q = 2t+1$, which is larger than
necessary, but simplifies the structure of the POVMs).  The total
subspace
\begin{eqnarray}
S \equiv \sp\left(\{\myket{S^t_j}\}_{j=0,\ldots,t} \otimes
\{\myket{c_h}\}_{h=0,1,\ldots,t} \otimes \{\ket{0}, \ket{1}\}\right)
\end{eqnarray}
supporting $\ket{\psi_{RS}(\phi,\theta)}$ breaks up into mutually
orthogonal subspaces $S_w$ of weight $w$, i.e., spanned by
computational basis states whose labels have Hamming weight $w$:
\begin{eqnarray}\label{eqn_S1}
S_1 &=& \sp\left(\myket{S^t_0}\otimes \{\ket{c_h}\}_h \otimes
\ket{0}\right)\\\label{eqn_Sk} S_k &=&
\sp\left(\myket{S^t_{k-1}}\otimes \{\ket{c_h}\}_h \otimes \ket{0},
\myket{S^t_{k-2}}\otimes \{\ket{c_h}\}_h \otimes
\ket{1}\right),\\
\label{eqn_St+2} S_{t+2} &=& \sp\left( \myket{S^t_t}\otimes
\{\ket{c_h}\}_h \otimes \ket{1}\right),
\end{eqnarray}
for $k=2,3,\ldots, t+1$.  Thus, for each $w$, we will do the
following:
\begin{itemize}
\item write $P_w$ in the basis in which $S_w$ is
expressed in Eqs (\ref{eqn_S1}), (\ref{eqn_Sk}), (\ref{eqn_St+2}),

\item derive an expression for $P_w  \myket{\psi_{RS}(0,0)}$ (which is proportional to $\ket{\Psi_w}$) in order to
find a basis for $\sp( \mathcal{O}_w) = \sp\{\ket{\Psi_w},V_\pi
\ket{\Psi_w}\}$ (which fully supports $E_{w,0}$), and

\item derive the form of $E_{w,0}$ and thus, by covariance, the form of the POVM
$\{E_{w,0},E_{w,\pi}\}$ in each subspace $S_w$.
\end{itemize}
\noindent Recalling Eq. (\ref{eqn_WLOG_RFstate}), it will be
convenient to let $\alpha_{j,h}\equiv b_j g_{j,h}$ and so
\begin{eqnarray}\label{eqn_repOfpsiR0}
\ket{\psi_R(0)} = \sum_{j,h} \alpha_{j,h} \myket{S^t_j} \myket{c_h}.
\end{eqnarray}

\vspace{3mm} \noindent \underline{$w=1$:}

Writing
\begin{eqnarray}
&&P_1 \ket{\psi_{RS}(0,0)}\\
&=& \left( \sum_{h} \myketbra{S^t_0}{S^t_0}\otimes
\ketbra{c_h}{c_h}\otimes
\ketbra{0}{0}\right)\ket{\psi_R(0)}{(\ket{0}+\ket{1})}/{\sqrt{2}}\\
&=&\myket{S^t_0}\left( \sum_h [(\mybra{S^t_0}\mybra{c_h}\ket{\psi_R(0)})/\sqrt{2}]\myket{c_h}\right)\ket{0}\\
&=&\myket{S^t_0}\left( \sum_h
[\alpha_{0,h}/\sqrt{2}]\myket{c_h}\right)\ket{0},
\end{eqnarray}
we see that $V_{\pi}\myket{\Psi_1} = \myket{\Psi_1}$ so that
$E_{1,0} = E_{1,\pi} = \myket{\Xi_{0}}\ket{0}\mybra{\Xi_0}\bra{0}$,
where $\myket{\Xi_{0}}$ is a state such that
\begin{eqnarray}
\myket{\Xi_{0}} \propto \myket{S^t_0}\sum_h
[\alpha_{0,h}/\sqrt{2}]\myket{c_h}.
\end{eqnarray}
We note that getting the outcome corresponding to this POVM element
does not give any information about $\theta$; we arbitrarily assign
a guess of ``$\theta = 0$'' to this outcome, without affecting
optimality (since $\theta$ is a priori uniformly distributed).

\vspace{3mm} \noindent \underline{$w \in \{2,3,\ldots,t+1\}$:}

Similarly, we can write
\begin{eqnarray}
&& P_w \myket{\psi_{RS}(0,0)}\\
&=&\myket{S^t_{w-1}}\left( \sum_h
[\alpha_{w-1,h}/\sqrt{2}]\myket{c_h}\right)\ket{0} +\\
&&\myket{S^t_{w-2}}\left( \sum_h
[\alpha_{w-2,h}/\sqrt{2}]\myket{c_h}\right)\ket{1}.
\end{eqnarray}
Chiribella et al. \cite{CAS05} show that $E_{w,0}$ may be assumed to
have rank 1 without loss of generality.  Thus $E_{w,0}$ may be
written $\myketbra{\eta_w}{\eta_w}$, where
\begin{eqnarray}
\ket{\eta_w}  &=& a \ket{\Xi_{w-1}}\ket{0} + b
\ket{\Xi_{w-2}}\ket{1},
\end{eqnarray}
for some complex coefficients $a$ and $b$, such that $|a|^2 + |b|^2
= 1$, where $\myket{\Xi_{w-1}}$ and $\myket{\Xi_{w-2}}$ are states
such that, for $j=0,1,\ldots,t$,
\begin{eqnarray}\label{def_Xis}
\myket{\Xi_{j}} \propto \sum_h
\frac{\alpha_{j,h}}{\sqrt{2}}\myket{S^t_j}\myket{c_h}.
\end{eqnarray}
We have (using covariance to get $E_{w,\pi}$)
\begin{eqnarray}
&& E_{w,0} + E_{w,\pi}\\
&=& 2(|a|^2\ket{\Xi_{w-1}}\ket{0}\bra{\Xi_{w-1}}\bra{0} +
|b|^2\ket{\Xi_{w-2}}\ket{1}\bra{\Xi_{w-2}}\bra{1}).
\end{eqnarray}
But
\begin{eqnarray}
&& E_{w,0} + E_{w,\pi}\\
&=&I_{\sp(\mathcal{O}_w)}\\
&=& \ket{\Xi_{w-1}}\ket{0}\bra{\Xi_{w-1}}\bra{0} +
\ket{\Xi_{w-2}}\ket{1}\bra{\Xi_{w-2}}\bra{1}.
\end{eqnarray}
Equating the two expressions implies that
\begin{eqnarray}
\ket{\eta_w} = \frac{1}{\sqrt{2}} (\ket{\Xi_{w-1}}\ket{0} + e^{i
\varphi_w}\ket{\Xi_{w-2}}\ket{1}),
\end{eqnarray}
for some phase $\varphi_w$.  But we must have $\varphi_w = 0$ since
$E_{w,0}$ corresponds to the guess ``$\theta = 0$''.

\vspace{3mm} \noindent \underline{$w=t+2$:}

Similar to the case $w=1$ and using the definition from Eq.
(\ref{def_Xis}), we have $E_{t+2,0}=E_{t+2,\pi} =
\myket{\Xi_t}\ket{1}\mybra{\Xi_t}\bra{1}$.  We assign the guess
``$\theta=\pi$'' to getting the outcome corresponding to this POVM
element.

To summarize, the elements of the overall POVM $\{E_0,E_\pi\}$
describing the measuring-and-guessing strategy may be expressed
\begin{eqnarray}
E_0 &=& \myket{\Xi_{0}}\ket{0}\mybra{\Xi_0}\bra{0}
+\sum_{w=2}^{t+1}\ketbra{w,+}{w,+} \\
E_\pi &=& \sum_{w=2}^{t+1}\ketbra{w,-}{w,-}
+\myket{\Xi_t}\ket{1}\mybra{\Xi_t}\bra{1},
\end{eqnarray}
where
\begin{eqnarray}
\myket{w,\pm} \equiv \frac{1}{\sqrt{2}} (\ket{\Xi_{w-1}}\ket{0} \pm
\ket{\Xi_{w-2}}\ket{1}).
\end{eqnarray}

\subsection*{A.3~~~Proof of Theorem \ref{thm_Main}, assuming Eq. (\ref{eqnPbreak})}

For security with error $\epsilon$, we require
\begin{eqnarray}
r(1 - c/(2r+1)^2)^s < \epsilon,
\end{eqnarray}
which, by taking the logarithm of both sides, is equivalent to
\begin{eqnarray}\label{ss}
s > \log(\epsilon/r)/\log(1 - c/(2r+1)^2).
\end{eqnarray}
Using the series expansion $\log(1 - x) = -(x + x^2/2 + x^3/3 +
\cdots)$, the right-hand side of Eq. (\ref{ss}) is upper-bounded by
\begin{eqnarray}
(2r+1)^2\log(r/\epsilon)/c,
\end{eqnarray}
from which the theorem follows.

\subsection*{A.4~~~Proof of Proposition \ref{prop_Maximization}}

This maximization problem is very similar to that in Ref.
\cite{BRST06}, where it was required to maximize $\mybra{\zeta}M'_t
\myket{\zeta}$ over all states $\ket{\zeta} \in \sp
\{\ket{j}:j=0,1,\ldots,t \}$ for
\begin{eqnarray}
M'_t = \sum_{j=0}^{t-1} \myketbra{j+1}{{j}} +\myketbra{{j}}{{j+1}}.
\end{eqnarray}
\noindent In fact, in light of Eq. (\ref{eqn_WLOG_RFstate}), the
phase estimation problem in Ref. \cite{BRST06} may be viewed as the
same as the one we consider, but where Eve does not have access to
the register $R'$.  (Indeed, our optimal success probability cannot
be less than that in Ref. \cite{BRST06}, since at the very least Eve
can forgo the use of the ancillary register $R'$.)  Finally, below,
we show that our optimal success probability is exactly equal to
that obtained in Ref. \cite{BRST06}.

Let $\alpha_{j,h}^\star$ denote the optimal values for our
maximization problem, and let $M_t^\star$,
$\myket{\psi_R(0)^\star}$, and $\myket{\Xi_{j}^\star}$ denote the
values of $M_t$, $\myket{\psi_R(0)}$, and $\myket{\Xi_{j}}$ at those
optimal values. Note that $\{\myket{\Xi_{j}}: j=0,1,\ldots,t\}$ is
orthonormal for all values of $\alpha_{j,h}$, thus
$\{\myket{\Xi_{j}^\star}: j=0,1,\ldots,t\}$ is orthonormal. Consider
now optimizing $\bra{\psi}M_t^\star\ket{\psi}$ over all unit vectors
$\ket{\psi} \in \sp\{\myket{\Xi_{j}^\star}: j=0,1,\ldots,t\}$ for
fixed $M_t^\star$; denote the optimal $\ket{\psi}$ as
$\ket{\psi^\star}$. It must be that
\begin{eqnarray}\label{eqn_ineq}
\bra{\psi^\star}M_t^\star\ket{\psi^\star} \geq
\bra{\psi_R(0)^\star}M_t^\star\ket{\psi_R(0)^\star},
\end{eqnarray}
since $\ket{\psi_R(0)^\star} \in \sp\{\myket{\Xi_{j}^\star}:
j=0,1,\ldots,t\}$ by inspecting Eqs (\ref{eqn_repOfpsiR0}) and
(\ref{def_Xis}).  Now note that the coefficients of
$\ket{\psi^\star}$ with respect to the basis
$\{\myket{\Xi_{j}^\star}: j=0,1,\ldots,t\}$ must be precisely those
coefficients of the optimal $\ket{\zeta}$ with respect to the
standard orthonormal basis $\{\ket{j}:j=0,1,\ldots,t \}$ found in
Ref. \cite{BRST06}; otherwise, substituting the coefficients of
$\ket{\psi^\star}$ would give a higher maximum than that in Ref.
\cite{BRST06}.  (The argument works because, in both cases, the
orthonormal basis is fixed for the optimization.)  Therefore, we
have, as in Ref. \cite{BRST06},
\begin{eqnarray}
\ket{\psi^\star} \propto \sum_{j=0}^t  \sin
\left[\frac{(j+1)\pi}{t+2}\right] \myket{\Xi_j}.
\end{eqnarray}

\end{document}